\documentclass[conference]{IEEEtran}

\ifCLASSINFOpdf
\else
\fi
\usepackage{amsmath, amsfonts, amsthm, amssymb, wasysym, txfonts, bm, bbm, graphicx}
\usepackage{multicol} 
\usepackage[font=small,labelfont=bf]{caption}
\usepackage{pdfpages, float}

\newtheorem{thm}{Theorem}[section]
\newtheorem{co}[thm]{Corollary}
\newtheorem{lem}[thm]{Lemma}

\newtheorem{assumption}[thm]{Assumption}

\newtheorem{pr}[thm]{Proposition}

\newtheorem{definition}[thm]{Definition}

\newtheorem{example}[thm]{Example}

\newtheorem{remark}[thm]{Remark}

\IEEEoverridecommandlockouts

\begin{document}
%
\title{A Generalized Algebraic Approach to Optimizing \\ SC-LDPC Codes}
\author{\thanks{This work is supported by NSF grants CCF-1440001, ECCS-1711056.}
\IEEEauthorblockN{Allison Beemer$^\ast$, Salman Habib$^\dagger$, Christine A. Kelley$^\ast$, and Joerg Kliewer$^\dagger$\\}
\IEEEauthorblockA{$^\ast$Department of Mathematics, University of Nebraska -- Lincoln, Lincoln, Nebraska 68588\\
}
\IEEEauthorblockA{$^\dagger$Department of Electrical and Computer Engineering, New Jersey Institute of Technology, Newark, New Jersey 07103 \\}
}


%


\maketitle

\begin{abstract} Spatially coupled low-density parity-check (SC-LDPC) codes are sparse graph codes that have recently become of interest due to their capacity-approaching performance on memoryless binary input channels. In this paper, we unify all existing SC-LDPC code construction methods under a new generalized description of  SC-LDPC codes based on algebraic lifts of graphs. We present an improved low-complexity counting method for the special case of $(3,3)$-absorbing sets for array-based SC-LDPC codes, which we then use to optimize permutation assignments in SC-LDPC code construction. We show that codes constructed in this way are able to outperform previously published constructions, in terms of the number of dominant absorbing sets and with respect to both standard and windowed decoding.
\end{abstract}


%
\IEEEpeerreviewmaketitle

\section{Introduction}
In recent years, it was shown that {\em spatially coupling} several copies of a Tanner graph of an LDPC code improves their density evolution (DE) thresholds and brings them closer to channel capacity \cite{KRU13}. This phenomenon, called {\em threshold saturation}, allows the SC-LDPC code to have the best threshold possible, i.e. the threshold under belief propagation (BP) using DE techniques approaches the maximum a posteriori (MAP) threshold. Further, it was shown that the threshold  approaches capacity as the degree of the nodes in the graph, the spatial coupling length, and the memory of the coupling increase. While these results are all asymptotic, it is desirable for practical applications to design finite length SC-LDPC codes that have better performance both in the waterfall and in the error floor regions compared to standard LDPC codes of comparable code rates, block lengths, and node degrees. 

Moreover, SC-LDPC  codes are suitable for windowed encoding and decoding in a streaming fashion which significantly reduces the latency compared to that of block codes. The BP decoding is performed on a window of variable and check nodes, and once these nodes are processed for some number of iterations, the window slides, and the nodes in the new window are processed \cite{PISWC10,ISUW13}.

 SC-LDPC codes are typically constructed by applying an edge-spreading process to a base Tanner graph. The resulting graph (called an SC-protograph) is often lifted to obtain the graph representation of the overall SC-LDPC code. For array-based SC-LDPC codes, a so-called cutting vector over an array-based block code may be used to determine the edge spreading connections. These methods will be reviewed in Section \ref{edge_spreading_lift}. The SC-protograph is critical in terms of obtaining SC-LDPC codes with good thresholds and good error floor. While the threshold behavior is controlled mainly by the memory, coupling length, and the degree of the nodes in the SC-protograph, the error floor behavior is heavily influenced by the absorbing sets in the SC-LDPC Tanner graph,
whose presence depends on the structure of the base graph, the edge-spreading method, and the permutations used in the terminal lift. Optimizing the cutting vector has been shown to remove harmful absorbing sets in the resulting code \cite{MDC14, ARKD15}. Moreover, in \cite{BK16}, the edge spreading process was modified to eliminate harmful trapping sets in the resulting SC-LDPC code.
In this paper, we present a new unified, single-step lifting method that performs both the edge-spreading and lifting steps of the SC-LDPC code construction. This method provides more flexibility in code construction, and provides an avenue to remove harmful absorbing sets algebraically via lifting. 

Note that the class of array-based LDPC (AB-LDPC) codes is a particular class of implementation-friendly, quasi-cyclic LDPC codes that have excellent performance, in particular for moderate block lengths \cite{F00}. In combination with spatial coupling the resulting codes inherit the excellent benefits of SC-LDPC codes highlighted above. 
We simplify the method of enumerating absorbing sets presented in \cite{AREKD16}, giving a line counting method of enumerating absorbing sets in array-based SC-LDPC (AB-SC-LDPC) codes. We use this method to find strategic choices of permutations in our general lift framework that can give codes outperforming those from optimized cutting vectors of AB-SC-LDPC codes. Furthermore, we demonstrate that our method yields a lower ratio of absorbing sets affecting a windowed decoder.

This work is organized as follows. Necessary background is given in Section \ref{prelims}. In Section \ref{edge_spreading_lift}, we show how common methods of designing SC-LDPC codes may be viewed as a single protograph construction with constraints on the algebraic lift. In Section \ref{removelift}, we discuss how absorbing sets may be removed algebraically using suitable choices of permutations. In Section \ref{counting}, we present a low complexity counting method for $(3,3)$-absorbing sets for AB-SC-LDPC codes, and in Section \ref{simulation} we provide results comparing several examples. Section \ref{conclusion} concludes the paper.

\section{Preliminaries}
\label{prelims}

In this section, we review the basic background for algebraic lifts of graphs, the protograph method of LDPC code construction, and absorbing sets (ABS). 

Let $[n] = \{1,2,\ldots, n\}$, and let
 $S_n$ denote the symmetric group on $n$ elements. That is, $S_n$ is the group of all permutations of $[n]$. Cycle notation for an element in $S_{n}$ is an ordering of the elements of $[n]$ in a list partitioned by parentheses, and is read as follows: each element is mapped to the element on its right, and an element before a closing parenthesis is mapped to the first element within that parenthesis. Each set of parentheses denotes a {\em cycle}.  The {\em cycle structure} of a permutation $\pi \in S_n$ is a vector $(c_1,\ldots, c_n)$ where, for $i \in [n]$,  $c_i$ denotes the number of $i$-cycles in the cycle notation of $\pi$.
We will also equate a permutation in $S_{n}$ with its corresponding $n\times n$ permutation matrix, where entry $(i,j)$ is equal to $1$ if $j\mapsto i$ in the permutation, and $0$ otherwise.

Let $G$ be graph with vertex set $V = \{v_1,\ldots, v_n\}$ and edge set $E \subseteq V \times V$. A {\em degree $J$ lift} of $G$ is a graph $\hat{G}$ with vertex set $\hat{V} = \{v_{1_1},\ldots, v_{1_J},\ldots, v_{n_1},\ldots, v_{n_J}\}$ of size $nJ$ and for each $e \in E$, if $e = v_iv_j$ in $G$, then there are $J$ edges from $\{v_{i_1},\ldots, v_{i_J}\}$ to $\{v_{j_1},\ldots, v_{j_J}\}$ in $\hat{G}$ in a one-to-one matching. To algebraically obtain a specific lift $\hat{G}$,  permutations may be assigned to each of the edges in $G$ so that if $e=v_{i}v_{j}$ is assigned a permutation $\tau$, the corresponding edges in $\hat{G}$ are $v_{i_k}v_{j_{\tau(k)}}$ for $1\leq k\leq J$. The edge $e$ is considered as directed for the purpose of lifting. Such edge assignments to the base graph, and the corresponding graph lift properties, are studied in \cite{GT01}. 

\emph{Protograph-based LDPC codes} are codes constructed from small Tanner graphs in this way \cite{T05}, where permutations are often chosen randomly. Without loss of generality, we assume all edges in a protograph are directed from variable node to check node for permutation assignments. We may also view this process as replacing nonzero entries of the protograph's parity-check matrix with $J\times J$ permutation matrices, when the lift is degree $J$. There are several methods for constructing SC-LDPC codes, including an edge-spreading protograph approach and the so-called cutting vector approach; we will describe how these constructions may be unified under a single graph lift framework in Section \ref{edge_spreading_lift}.

Combinatorial structures in the Tanner graph, such as absorbing sets, have been shown to cause iterative decoder failure. An \emph{$(a,b)$-absorbing set} is a subset $D$ of variable nodes such that $|D|=a$, $|O(D)|=b$, and each variable node in $D$ has strictly fewer neighbors in $O(D)$ than in $F\setminus O(D)$, where $N(D)$ is the set of the check nodes adjacent to variable nodes in $D$, $O(D)$ is the subset of check nodes of odd degree in the subgraph induced on $D\cup N(D)$, and $F$ is the set of all check nodes \cite{MDC14}. In Sections IV-VI, we optimize permutation assignments to minimize the number of harmful ABS in an SC-LDPC code.

\section{SC-LDPC Codes from Algebraic Lifts}\label{edge_spreading_lift}

In the protograph approach to SC-LDPC code construction, a base Tanner graph is copied and {\em coupled} to form the {\em SC-protograph}. There are many ways to couple the edges from one copy of the base graph to the other copies; this process of coupling is generally termed {\em edge-spreading}. The SC-LDPC code is then defined by a {\em terminal lift} of the resulting SC-protograph.

\subsection{SC-LDPC codes via edge-spreading}
\label{edgespreading}

To construct the SC-protograph, $L$ copies of a base graph, such as the one shown in Fig. \ref{basegraph}, are {\em coupled}. The coupling process may be thought of as first replicating the base graph at positions $0,\ldots, L-1$, and then ``edge-spreading" the edges to connect the variables nodes at position $i$ to check nodes in positions $i, \ldots, i+m$ so that the degrees of the variable nodes in the base graph are preserved.  The number $L$ of copies of the base graph is referred to as the {\em coupling length}, and the number $m$ of future copies of the base graph that an edge may spread to is called the \emph{memory} or \emph{coupling width}. In the case of a \emph{terminated} SC-protograph, terminating check nodes are introduced at the end of the SC-protograph  as necessary to {\em terminate} the SC-protograph. An example of an SC-protograph obtained by coupling the base graph in Fig. \ref{basegraph} is given in Fig. \ref{terminatedSC}. Allowing edges to instead loop back around to the first few positions of check nodes results in a \emph{tailbiting} SC-protograph, in which both variable and check node degrees are preserved in all positions. 

\begin{figure}[h]
\centerline{\resizebox{2.0in}{0.7in}{\includegraphics{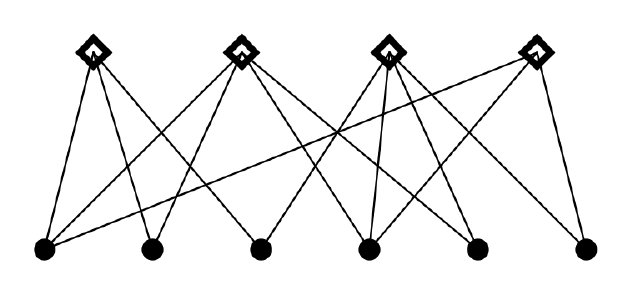}}}
\caption{Base Tanner graph to be coupled to form an SC-protograph. Variable nodes are denoted by \textbullet, and check nodes are denoted by $\Diamond$.}
\label{basegraph}
\end{figure}

\vspace{-0.2in}

\begin{figure}[h]
\centerline{\resizebox{3.6in}{0.8in}{\includegraphics{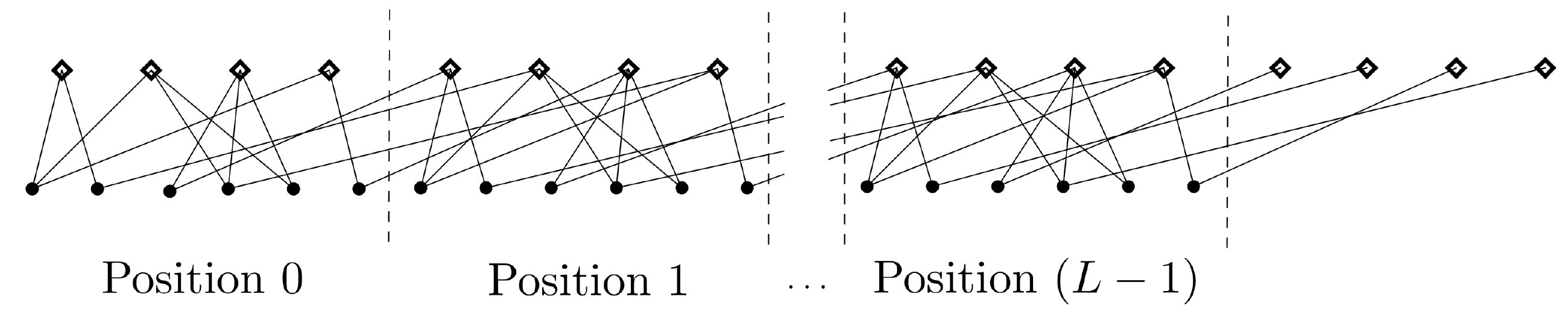}}}
\caption{Terminated SC-protograph resulting from randomly edge-spreading $L$ copies of the Tanner graph in Fig.~\ref{basegraph} with memory $m = 1$, and applying the same map at each position.}
\label{terminatedSC}
\end{figure}

This edge-spreading process may also be viewed in terms of the parity-check matrix, $H$, of the base graph. Edge-spreading is equivalent to splitting $H$ into a sum of $m+1$ matrices of the same dimension, so that $H=H_{0}+H_{1}+\cdots+H_{m}$, and then arranging them as in Matrix (\ref{terminatedmatrix}) below to form the parity-check matrix of a terminated SC-protograph with $L$ block columns. The tailbiting code corresponding to this terminated code has parity-check matrix as in Matrix (\ref{tailbitingmatrix}), so that every check node has degree equal to its corresponding vertex in the base graph.

\vspace{-0.2in}

\begin{multicols}{2}
{\scriptsize
\begin{equation}
\label{terminatedmatrix}
\begin{pmatrix}
H_{0} & & & &\\
H_{1} & H_{0} & & & \\
\vdots & \vdots & \ddots & &\\
H_{m} & & & &\\
& H_{m} & & &  \\
& & & &  H_{0}  \\
& & & \ddots & \vdots \\
& & & & H_{m}
\end{pmatrix}
\end{equation}
\vfill\null
\columnbreak
\begin{equation}
\label{tailbitingmatrix}
\begin{pmatrix}
H_{0} & &&  H_{m}\cdots & H_{1} \\
H_{1} & \ddots & & \ddots &  \vdots\\
\vdots & & & & H_{m}\\
H_{m} & & &  \ddots &\\
& \ddots & &   &  \\
&  & H_{m} & \cdots &  H_{0}  \\
\end{pmatrix}
\end{equation}
}
\end{multicols}

\vspace{-0.2in}

Edge-spreading may be done in a variety of ways. Two common methods are \cite{KRU13}: (i) For each variable node $v$ in Position 0, if $v$ has $j$ neighbors $c_1, \ldots, c_j$ in the base graph, randomly choose for each $\ell  = 1, \ldots, j$, a copy of $c_{\ell}$ from the Positions $0, \ldots, m$, and 
(ii) if each variable node in Position 0 has $j$ neighbors in the base graph, randomly choose $j$ of the Positions $0, \ldots, m $ to spread edges to, then, for each of the $j$ neighbors $c_1, \ldots, c_j$ of a variable node $v$, randomly choose a check neighbor from the copies of $c_{\ell}$ ($\ell  = 1, \ldots, j$) such that $v$ has exactly one neighbor in each of the chosen $j$ positions, and exactly one of each check node neighbor type. Note that method (ii) is a special case of (i).

 In the case of array-based codes, edge-spreading is typically accomplished using a \emph{cutting vector} \cite{MDC14}. Requiring a memory of $m=1$, the cutting vector describes the split of the base matrix into $H_{0}$ and $H_{1}$ via a diagonal cut. In particular, for an array-based parity-check matrix with $\gamma$ block rows, the cutting vector is denoted $\xi=(\xi_{0},\xi_{1},\xi_{2},\ldots \xi_{\gamma-1})$; in block row $i$, 
 block columns $j$ for $0\leq j <\xi_{i}$ are placed in $H_{0}$. The remaining block entries above this cut belong to $H_{1}$.
This approach has been expanded in \cite{MR17} and \cite{EHD17} to allow for higher memory and more freedom in the edge-spreading structure, though blocks of edges remain spread as single units.

Regardless of the edge-spreading method, the mapping given at Position 0 will be applied at future positions $1,\ldots,L-1$;  a \textit{terminal lift} may then be applied to the SC-protograph, yielding the SC-LDPC code. Repeating the edge-spreading at future positions allows the resulting SC-LDPC code to be a terminated LDPC convolutional code if the permutations applied to lift the resulting SC-protograph are cyclic permutations \cite{TSSFC04, PSVC11}. In general, terminated SC-LDPC 
codes are desirable for practical applications \cite{MDC14, ARKD15, TSSFC04}.

\subsection{Viewing edge-spreading algebraically}

In the remainder of this section, we describe how the edge-spreading process may be viewed as an approximate graph lift. 
First, we note that to construct a terminated SC-protograph, we may break a tailbiting protograph, copying the constraint nodes at which the graph is broken. We claim that a tailbiting SC-protograph may be viewed as a degree $L$ lift of the base graph -- where $L$ denotes the coupling length -- by considering the $L$ copies of a node type in the SC-protograph to be the lift of the corresponding node in the base graph. While a terminated SC-protograph is not, then, strictly a graph lift of the base graph, the set of terminated SC-protographs is in one-to-one correspondence with the set of tailbiting SC-protographs, and so each can be associated with a lift of the base graph.

Recall that once an edge-spreading assignment is made for variable nodes in a single position, that same edge-spreading is repeated at all future positions. This translates to the following:

\begin{lem}
\label{restrictions}
To construct a tailbiting SC-protograph with coupling length $L$ and memory $m$ from a base graph via a graph lift, the possible permutation edge assignments from the permutation group $S_{L}$ are the permutations corresponding to $\tau_{L}^{k}$, for $0\leq k\leq m$, where $\tau_{L}$ is the $L\times L$ identity matrix left-shifted by 1, and at least one assignment corresponds to $\tau_{L}^{m}$.
We denote this set of permutations by $A_{L,m}$.
\end{lem}

\begin{proof}
The proof follows from the structure of the tailbiting SC-protograph, as given in Matrix (\ref{tailbitingmatrix}) above.
\end{proof}

\begin{example}
Suppose $L=6$ and $m=3$. Then,
$A_{6,3}=\{\tau_{6}^{0},\ldots,\tau_{6}^{3}\}$. If $L=7$ instead, $A_{7,3}=\{\tau_{7}^{0},\ldots,\tau_{7}^{3}\}$.
\end{example}

Given a fixed memory, we may spread edges by simply assigning allowed permutations to edges in the base graph uniformly at random. This is equivalent to method (i) of edge-spreading, as described in Section \ref{edgespreading}. Method (ii) is more restrictive: it stipulates that for a given variable node, each possible permutation assignment is used at most once on its incident edges. 

This framework may be applied to a variety of existing methods for coupling, with additional restrictions on possible permutation assignments in each case. In Section \ref{nested_inclusions}, we will discuss how it may be used to decribe the cutting vector, as well as the generalized cutting vectors of \cite{MR17} and \cite{EHD17}. 

To arrive at the standard matrix structure of the terminated SC-protograph as given in Matrix (\ref{terminatedmatrix}) -- and hence the correct ordering of bits in a codeword --, one must rearrange the rows and columns of the matrix resulting from this lift: 
each $L\times L$ block that has replaced an entry in the base parity-check matrix corresponds to edges of a single type (i.e. between a single type of variable and check) in the SC-protograph. To have the ordering of variable and check nodes in Matrix (\ref{tailbitingmatrix}), we should place the first variable node of type 1 with the first variable node of type 2, etc., and similarly with check nodes. 
That is, ordering is done primarily by a vertex's index within an $L\times L$ block (ranging from $1$ to $L$), and secondarily by the index of that $L\times L$ block (ranging from $1$ to $V$, where $V$ is the number of columns of the base matrix). Rearranging rows and columns does not change the structure of the associated graph (e.g. the minimum distance of the underlying code, or the number of ABS therein), but places bits in the correct order, highlights the repeated structure, and allows us to break the tailbiting portion of the code and yield a parity-check matrix in the form of Matrix (\ref{terminatedmatrix}). In particular, this is useful for a sliding windowed decoder. 

\subsection{Combining edge-spreading and the terminal lift}

Edge-spreading and the terminal lift may be combined into a single, higher-degree lift. In other words, the entire construction process of first replacing each nonzero entry of the base matrix with an $L\times L$ circulant matrix of the form $\tau_{L}^{k}$, and then replacing each nonzero entry of each $\tau_{L}^{k}$ with an unrestricted $J\times J$ permutation matrix $\lambda$ to perform the terminal lift, may be accomplished in a single step by assigning permutations from $S_{JL}$ to edges in the base graph.
Making a single assignment per edge of the base graph, and thus per edge of the same type in the SC-protograph, is useful for two reasons: (1) breaking ABS in the base graph will break ABS in the terminally-lifted Tanner graph, and (2) the structure of the code is repeated, reducing storage and implementation complexity, particularly for windowed decoding.

\begin{thm}
\label{doublerestrictions}
To construct a tailbiting SC-LDPC code with coupling length $L$, memory $m$, and terminal lift of degree $J$ from a base graph via a single graph lift, the possible permutation edge assignments from the permutation group $S_{JL}$ are those whose corresponding matrix is of the form $ \tau_{L}^{k} \otimes \lambda $
where ``$\otimes$" denotes the Kronecker product, $\tau_{L}$ is the $L\times L$ identity matrix left-shifted by $1$, $0\leq k\leq m$, and $\lambda$ is any $J\times J$ permutation matrix. We denote this set of permutations by $B_{L,m,J}$.
\end{thm}

\begin{proof}
The proof is clear from Lemma \ref{restrictions} and the above discussion.
\end{proof}

Notice that for $J=1$, $B_{L,m,J}=A_{L,m}$. To give the parity-check matrix of the SC-LDPC code the structure of Matrix (\ref{terminatedmatrix}), we must again rearrange rows and columns after this lift is performed, and then break the tailbiting code to form a terminated SC-LDPC code. In this case, however, rows and columns are rearranged as blocks, so that $J\times J$ blocks corresponding to choices of $\lambda$ remain intact.

\begin{thm}
\label{spreadliftsubgroup}
The set $B_{L,(L-1),J}$ has size $(m+1)\cdot J!$, and the element $\tau_{L}^{k} \otimes \lambda$ has order 
\[\frac{L\cdot o(\lambda)\cdot gcd(k,L,o(\lambda))}{gcd(k,L)\cdot gcd\left(L,o(\lambda)\right)}\]
\noindent where $o(\lambda)$ indicates the order of the permutation $\lambda$. Furthermore, $B_{L,(L-1),J}$ forms a subgroup of $S_{JL}$ for any choice of $J$ and $L$.
\end{thm}

\begin{proof}
The proof follows from the cycle structure of $\tau_{L}^{k}$ 
 and properties of the Kronecker product.

\vspace{-0.1in}
\end{proof}

We now discuss the case where we restrict the permutation $\lambda$ to be a cyclic shift of the $J\times J$ identity matrix.  

\begin{co}
The permutation given by $\tau_{L}^{k} \otimes \tau_{J}^{\ell}$ has order 
\[ \frac{JL\cdot gcd(k, J, L) \cdot gcd(\ell, J, L)}{gcd(\ell, J)\cdot gcd(k,L) \cdot gcd(J,L) \cdot gcd(k,\ell,J,L)}. \]
\end{co}

If we use permutations of this type to lift a base matrix, we may say more about the structure of the parity-check matrix of the resulting SC-LDPC code.

\begin{lem}
\label{quasicyclic}
If a base parity-check matrix is lifted to form an SC-LDPC code using permutation matrices of the form $\tau_{L}^{k} \otimes \tau_{J}^{\ell}$ in $B_{L,m,J}$, for $J\geq 2$, then the resulting parity-check matrix is quasi-cyclic, independently of whether block rows and columns are reordered.
\end{lem}

This structure is a consequence of the terminal degree $J\geq 2$ lift. 
Note that when $J=1$ (i.e. the SC-protograph is not lifted), the parity-check matrix is not necessarily quasi-cyclic post-reordering, but will be if the base matrix is array-based and permutations are assigned constantly on blocks.

\subsection{Comparison of construction methods}\label{nested_inclusions}

Of the existing methods for SC-LDPC code construction, the framework presented in Theorem \ref{doublerestrictions} is the most general. In particular, traditional cutting vectors and the generalized cutting vector constructions of \cite{MR17} and \cite{EHD17} form a proper subset of this approach. 

Given a fixed array-based base graph, let the set of SC-LDPC codes formed with all possible edge-spreadings and terminal lifts as described in Theorem \ref{doublerestrictions} be given by $A$, the set of codes formed using a traditional cutting vector (without a terminal lift) be given by $C$,
the set of codes formed using a generalized cutting vector (also without a terminal lift) be given by $C_{g}$,
the set of codes for which there is no terminal lift ($J=1$ in Theorem \ref{doublerestrictions}) be given by $E$,
and the set of codes formed by restricting $\lambda$ of Theorem \ref{doublerestrictions} to be of the form $\tau_{J}^{\ell}$ (as in Lemma \ref{quasicyclic}) be given by $Q$. Then we have the following nested set inclusions:

\begin{pr} With $C$, $C_{g}$, $E$, $Q$, and $A$ defined as above,
\[ C \subsetneq C_{g} \subsetneq E \subsetneq Q \subsetneq A.\]
\end{pr}

For an SC-LDPC code constructed using a traditional cutting vector in $C$, the memory is equal to 1, and so Lemma \ref{restrictions} stipulates that the two possible permutation assignments to edges of the base graph are the identity and $\tau_L$. However, there is additional structure: if the cutting vector is $\xi=(\xi_{0},\xi_{1},\ldots)$, 
then the first consecutive (blocks of) $\xi_{0}$ variable nodes have the identity assigned to all of their edges, while the next $\xi_{1}$ consecutive (blocks) have the permutation $\tau_L$ assigned to their first edge\footnote{The ordering on the edges incident to a variable node is induced by the ordering of the corresponding parity-check matrix.}, and the identity assigned to all later edges, the next $\xi_{2}$ (blocks of) variable nodes have the permutation $\tau_L$ assigned to their first two edges, and the identity to all later  edges, etc.

In the generalized cutting vector approach of \cite{MR17} and \cite{EHD17}, which are edge-spreading methods applied specifically to array-based codes, blocks in the base matrix have constant assignments, but assignments to those blocks are not as restricted as in the traditional approach. 
The Minimum Overlap (MO) partitioning of \cite{EHD17} minimizes the number of edges incident to a given variable node that are assigned the same permutation.

Relaxing the restriction of constant assignments per block of an array-based code and allowing multiple permutation assignments per block is enough to show $C_{g} \subsetneq E$.
The final two inclusions are clear from Theorem \ref{doublerestrictions} and Lemma \ref{quasicyclic}.

\section{Removing Absorbing Sets}
\label{removelift}

It has been shown that for protograph-based LDPC codes, substructures such as trapping or stopping sets may be removed and girth may be improved with certain permutation assignments in the lifting process \cite{K08,ICV08}. Consequently, we may remove remove absorbing sets (ABS) by choosing suitable permutation assignments when constructing an SC-LDPC code via Theorem \ref{doublerestrictions}.

As an example, we will consider  base graphs which are array-based column-weight 3 codes of the form 

\begin{equation}
\label{H3p}
H(3,p)=
\begin{pmatrix}
I & I & I & \cdots & I \\
I & \sigma & \sigma^{2} & \cdots & \sigma^{p-1} \\
I & \sigma^{2} & \sigma^{4} & \cdots & \sigma^{2(p-1)}
\end{pmatrix},
\end{equation} 

\noindent where $\sigma$ is the $p\times p$ identity matrix left-shifted by 1.
Fig. \ref{3342AS} shows (3,3)- and (4,2)-ABS, which have been shown to be the most harmful to error floor performance in such codes \cite{AREKD16}. 
Notice that there is a 6-cycle in each of these ABS. To remove them algebraically by lifting we can assign permutations to the edges of the cycle that increase the cycle lengths corresponding to those edges.  This may be done using the following known result, where the \emph{net permutation} of a path is the product of the oriented edge labels. 

\begin{thm}\cite{GT01}
\label{voltageproduct}
If $C$ is a cycle of length $k$  with net permutation $\pi$ in the graph $G$, and $\hat{G}$ is a degree $J$ lift of $G$, then the edges corresponding to $C$ in $\hat{G}$ will form $c_1 + \cdots + c_J$ components with exactly $c_i$ cycles of length $ki$ each, where $(c_1,\ldots,c_J)$  is the cycle structure of $\pi$.
\end{thm}

\vspace{-0.2in}

\begin{figure}[h]
\centering
\includegraphics[scale=0.65]{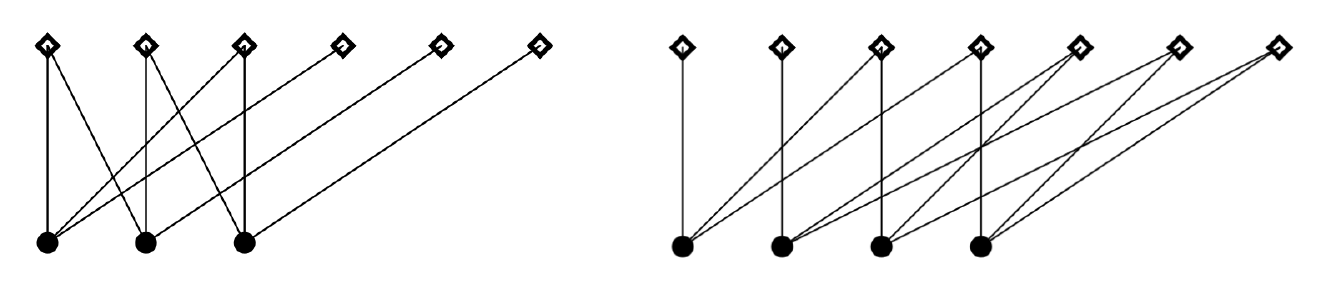}
\caption{A $(3,3)$- and a $(4,2)$-ABS in a column-weight 3, array-based code. Variable nodes are denoted by \textbullet, and check nodes are denoted by $\Diamond$.}
\label{3342AS}
\end{figure}

In the case of SC-LDPC codes, the permutation assignments are limited to those detailed in Theorem \ref{doublerestrictions}. However, even if we restrict ourselves to $m=1$ (or $2$) and no terminal lift, assigning the permutation $\tau_L$ (and $\tau_{L}^{2}$) to a strategic subset of the edges of a $6$-cycle will break the $6$-cycle in the lift, and hence will break the corresponding $(3,3)$-ABS.  Notice that since the $(3,3)$-ABS is a subgraph of the $(4,2)$-ABS in Fig. \ref{3342AS}, the latter are also removed.  
This motivates the algorithmic approach for optimizing permutation assignments in Sections \ref{counting} and \ref{simulation}, where we will focus on the case where edge assignments are made per block (as in \cite{MR17} and \cite{EHD17}). However, this restriction may be relaxed and multiple assignments made per block, which will be illustrated in the full version.

\section{Counting $(3,3)$-absorbing sets}
\label{counting}

In this section we present a novel line counting approach to the problem of
enumerating $(3,3)$-ABS. This work is a simplification of
the approach in \cite{AREKD16}, which is based on counting integer points
within a polygon. Note that the enumeration technique discussed in
\cite{AREKD16} applies only to AB-SC-LDPC codes obtained by the cutting vector
approach; however, our line counting method is applicable for enumerating
$(3,3)$-ABS in any column-weight 3 AB-LDPC code. 
The method emanates from the
structural properties of the 6-cycles associated to $(3,3)$-ABS in an AB-LDPC code. From the cyclic structure of these codes, it is straightforward
to show that a 6-cycle in an AB code implies the presence of a $(3,3)$-ABS
\cite{AREKD16}.

We enumerate $(3,3)$-ABS in two types of AB-SC-LDPC codes -- those obtained by the cutting vector approach, whose parity-check matrices are denoted as
$H(3,p,\boldsymbol{\xi},L)$, and those obtained via the general algebraic lifting of Section \ref{edge_spreading_lift},
with parity-check matrix $H(3,p,L)$. Note that $H(3,p)$, as in Equation (\ref{H3p}), is the base matrix in both
cases. 
Let $B_m$ denote a $(m+1)$-ary permutation assignment matrix of dimension
$\gamma\times p$, which will determine $H(3,p,L)$.
An entry $M\in\{0,\ldots,m\}$ in position $(i,j)$ of $B_{m}$ indicates that all non-zero elements of block
$(i,j)$ of $H(3,p)$ should be lifted by by $\tau_{L}^{M}$.
In Section~\ref{simulation}, we obtain $B_m$ via numerical
optimization.

A 6-cycle in an AB code may be indicated by the six (row, column) pairs associated with its edges in the corresponding Tanner graph, using the notation $(r_{t},c_{\ell})$ for $t,\ell\in[3]$. Note that not all index combinations are possible; indeed, two consecutive edges in a cycle must share either a row or a column index. Let $q_{t}$ and $s_{t}$ denote the block row index and the row index within a particular block row, respectively, such that $r_{t}=q_{t}p+s_{t}$. Similarly, let $j_{\ell}$ and $k_{\ell}$ represent the block column index and the column index inside a particular block column, so that $c_{\ell}=j_{\ell}p+k_{\ell}$. In this way, the location of a vertex $(r_{t},c_{\ell})$ may be written as $(q_t, s_t, j_\ell, k_\ell)$, $t, \ell \in [3]$. This may be seen in Fig. \ref{fig:6-cycle}.

\begin{figure}[h]
\centering
\includegraphics[scale=0.5]{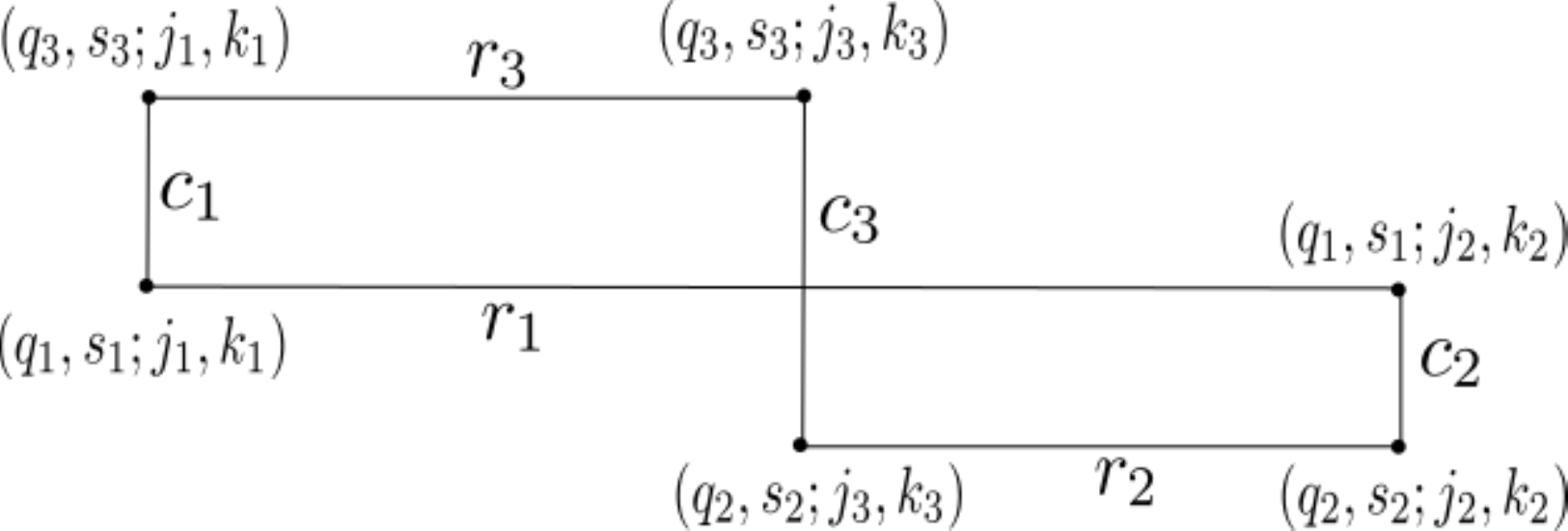}
\caption{The structure of a 6-cycle within an AB code.}
\label{fig:6-cycle}
\end{figure}

Due to the structure of an AB-SC-LDPC code, 
 any 6-cycles present will span 3 distinct block rows and $p$ or $(m+1)p$ block columns. These three block rows will each have one of the following structure types: 

\begin{itemize}
\item[] \textbf{Type 1:} Consists of only identity matrices.
\item[] \textbf{Type 2:} Consists of matrices $\sigma^{z}$ for $0\leq z\leq p-1$.
\item[] \textbf{Type 3:} Consists of matrices
$\sigma^{2z}$ for $0\leq z\leq p-1$.
\end{itemize}

From this and other structural observations, we may obtain:

\begin{lem}
\label{lem:1} 
A $(3,3)$-absorbing set can only exist in a region of an AB-SC-LDPC code that
consists of three consecutive block rows, where the block columns within these block rows have weight at least two. In particular, there must be at least one block row in the $(3,3)$-absorbing set whose nonzero blocks are all identity matrices.
\end{lem}

\begin{lem}
\label{lem:3}
Let $c_{1}$ and $c_{2}$ denote the two columns of the 6-cycle which belong to the block row comprised of identity matrices,
with
$c_{2}>c_{1}$. 
Then, for some $n\in [p-1]$,
\begin{equation}
\label{eq:np}
c_{2}-c_{1}=np.
\end{equation}
\end{lem}

While a 6-cycle may be uniquely identified by
its column indices, two 6-cycles may share two column indices and differ in
the third. 
Moreover, the range of values of the third column, $c_{3}$ may be expressed in terms of the first two. Bounding the range of possible $c_{3}$ values using other structural observations, we may produce diagonal boundaries in the $(c_{1},c_{2})$ plane. The structure of AB-LDPC codes imposes additional boundaries on the range of values of $c_{1}$ and $c_{2}$. Let $\mathcal{R}$ denote a region of the AB-SC-LDPC code (with $\gamma =3$) containing the three distinct block row types discussed previously, and $p$ or $(m+1)p$ block columns. The number of 6-cycles existing within such a region is then proportional to the length of the segment of the line  (\ref{eq:np}) enclosed within 
the boundaries in the $(c_1,c_2)$ plane induced by the array-based structure.
This dramatically simplifies the enumeration of \cite{AREKD16}.

Note that from the cycle structure shown in Fig. \ref{fig:6-cycle}, the block column index $j_{3}$ must have columns with non-zero elements in both block row indices $q_{2}$ and $q_{3}$. Let the smallest and largest index of the block columns satisfying this property be $\alpha$ and $\beta-1$, respectively, where $\alpha<\beta$, $0\leq\alpha\leq p-1$, and $1\leq\beta\leq p$. In particular, $\alpha \leq j_{3} \leq \beta-1$. 

\begin{lem}
\label{lem:5}
The following inequalities hold for $c_{1}$ and $c_{2}$: 
\begin{align}
&\text{Case $1$:} &\quad &\frac{\alpha p}{2}\leq
                        c_{2}-\frac{1}{2}c_{1}<\frac{\beta
                        p}{2}\label{eq:2}\\
&\text{Case $2$:} & \quad & \frac{p^{2}+\alpha p}{2}\leq
                         c_{2}-\frac{1}{2}c_{1}<\frac{p^{2}+\beta
                         p}{2}\label{eq:3}\\
& \text{Case $3$:} & \quad & p^{2}-\beta p<c_{2}-2c_{1}\leq p^{2}-\alpha
                          p\label{eq:4} \\
& \text{Case $4$:} & \quad & -\beta p<c_{2}-2c_{1}\leq-\alpha p\label{eq:5}
\end{align}

\end{lem}

Lemma \ref{lem:5} can be explained as follows: when $c_{1}$ 
is connected to type-1 and type-2 
block rows, the conditions $c_{3}<p^{2}$ and
$c_{3}\geq p^{2}$ generate (\ref{eq:2}) and (\ref{eq:3}), respectively. When $c_{1}$ 
is connected to a type-1 and type-3 
block rows, (\ref{eq:4}) and (\ref{eq:5}) arise when $c_{3}<p^{2}$ and $c_{3}\geq p^{2}$, respectively. Recall that $c_{3}$ exists only between type-2 and type-3 block rows. The positioning of the circulant matrices of $\mathcal{R}$  places additional constraints on the range of values of $c_{1}$ and $c_{2}$. Let
\begin{align}
w_{1}p\leq c_{1}&<w_{2}p\label{eq:c1} \quad \text{and} \\
w_{3}p\leq c_{2}&<w_{4}p\label{eq:c2},
\end{align}
where $w_{1},w_{2},w_{3},w_{4}$ are integers satisfying $0\leq w_{1}\leq
 p-2$, $1\leq w_{2}\leq p-1$, $w_{1}+1\leq w_{3}\leq p-1$, $w_{2}+1\leq
 w_{4}\leq p$. Also, note that $n$ of the line in \ref{lem:3} must be contained in $\{1,\ldots,w_4-w_1-1\}$. The upper and lower bounds from  (\ref{eq:c1})
 (resp., equation (\ref{eq:c2})) produce vertical (resp., horizontal) boundaries on the $(c_{1},c_{2})$ plane. Since Cases $1,2,3,$ and $4$ in Lemma~\ref{lem:5} are mutually exclusive, the following theorem is obtained:

\begin{thm}
The number of $(3,3)$-absorbing sets within $\mathcal{R}$ is equal
to the sum of the number of $(c_1,c_2)$ integer pairs obtained from each of the cases in Lemma~\ref{lem:5}. 
\end{thm}

Recall that $\alpha$ (resp., $\beta-1$) is the lower (resp., upper) bound
on the range of possible block column indices of $c_{3}$; similarly, $w_{1}$ (resp., $w_{2}-1$) and $w_{3}$
(resp., $w_{4}-1$) are the lower (resp., upper) bounds on the column numbers  of $c_{1}$
and $c_{2}$, respectively.
Consequently, $S_\ell:=\{\alpha,\beta,w_1,w_2,w_3,w_4\}_\ell$ is the set
of input parameters for the line counting algorithm for Case $\ell \in [4]$. 

In the next subsection, we derive an analytical expression for the number of
$(3,3)$-ABS of an AB-SC-LDPC code via a line counting algorithm.
We then apply this approach to $(3,3)$-ABS for
$H(3,p,\boldsymbol{\xi},L)$ in Subsection \ref{B}, and extend the technique to \emph{piecewise} line counting to enumerate $(3,3)$-ABS in the more general case of $H(3,p,L)$ in Subsection \ref{C}.

\subsection{Enumeration of $(3,3)$-absorbing sets in $\mathcal{R}$ via line counting}
Let $\mathcal{N}_{\ell,\mathcal{R}}$ be the number of integer points
$(c_{1},c_{2})$ on the line in ($\ref{eq:np}$) that satisfy the conditions
for Case $\ell\in [4]$ in Lemma \ref{lem:5}. 
For $\mathcal{N}_{1,\mathcal{R}}$, consider the following lines:
let $l_{1a}$ (resp., $l_{2a}$) be the line obtained from the lower (resp., upper) bound of (\ref{eq:2}), $l_3$ (resp., $l_4$) the line obtained from the lower (resp., upper) bound of (\ref{eq:c1}), $l_5$ (resp., $l_6$) the line obtained from the lower (resp., upper) bound of (\ref{eq:c2}), and $l_7$ the line in (\ref{eq:np}). 
For example, $l_{1a}$ represents $c_{2}-\frac{1}{2}c_{1}=\frac{\alpha p}{2}$. These lines are shown in
Fig. \ref{fig:N1R}. 
Note that a $(c_1,c_2)$ integer
pair on $l_7$ in the grey region of Fig. \ref{fig:N1R} indicates an existing 6-cycle (and hence a
$(3,3)$-ABS) for the case $\ell=1$. 

Let $\theta_{i}$ be the point of intersection between $l_{7}$ and $l_{i+2}$ for $i\in [4]$, and let $\phi_{1a}$ (resp., $\phi_{2a}$) be the point of intersection between $l_{7}$ and $l_{1a}$ (resp., $l_{2a}$).
 Note that $\theta_1$ and $\theta_3$ are obtained from the lower bounds of
 (\ref{eq:c1}) and (\ref{eq:c2}), respectively, and hence they lie on the
 lower left corner of the region; by similar reasoning,
 $\{\theta_{2},\theta_{4}\}$ may be found on the upper right corner. The two points (one picked from each set) producing the
 shortest length of $l_{7}$ within the rectangular boundary imposed by
$l_{3},l_{4},l_{5},l_{6}$, are the points of interest. These points are
 denoted by $\sigma_{1a}$ and $\sigma_{2a}$. That is,
 $\sigma_{1a,x}=\max(\phi_{1a,x},\theta_{1x},\theta_{3x})$,
 $\sigma_{1a,y}=\max(\phi_{1a,y},\theta_{1y},\theta_{3y})$,
 $\sigma_{2a,x}=\min(\phi_{2a,x},\theta_{2x},\theta_{4x})$,
 $\sigma_{2a,y}=\min(\phi_{2a,y},\theta_{2y},\theta_{4y})$. Moreover, the
 y-coordinates of the points of intersection between $l_{3}$, and $l_{2a}$,
 $l_{4}$ and $l_{1a}$ are $(w_{1}+\beta)p/2$ and $(w_{2}+\alpha)p/2$,
 respectively. Then, from the principles of Cartesian geometry and the
 constraints given by (\ref{eq:2}), (\ref{eq:c1}) and (\ref{eq:c2}), we  obtain
\begin{equation}
\label{eq:N1R}
\mathcal{N}_{1,\mathcal{R}}=
\begin{cases}
\sum_{n=1}^{w_{4}-w_{1}-1}\left(\frac{\sqrt{(\sigma_{2a,x}-\sigma_{1a,x})^{2}+(\sigma_{2a,y}-\sigma_{1a,y})^{2}}}{\sqrt{2}}\right),
& \text{if } (\ast) \\
0, & \hspace{-9ex}\text{otherwise,}
\end{cases}
\end{equation}

\noindent where ($\ast$) denotes the conditions $\theta_{1y}<\frac{(w_{1}+\beta)p}{2}$,
$\theta_{2y}>\frac{(w_{2}+\alpha)p}{2}$, $\theta_{1x}\leq\{\sigma_{1a,x},\sigma_{2a,x}\}\leq\theta_{2x}$, and $\theta_{1y}\leq\{\sigma_{1a,y},\sigma_{2a,y}\}\leq\theta_{2y}$.
$\mathcal{N}_{2,\mathcal{R}},\mathcal{N}_{3,\mathcal{R}}$ and $\mathcal{N}_{4,\mathcal{R}}$ can be obtained in similar fashion from (\ref{eq:3}), (\ref{eq:4}) and (\ref{eq:5}), respectively.

\begin{figure}[htb]
\centering
\includegraphics[scale=0.5]{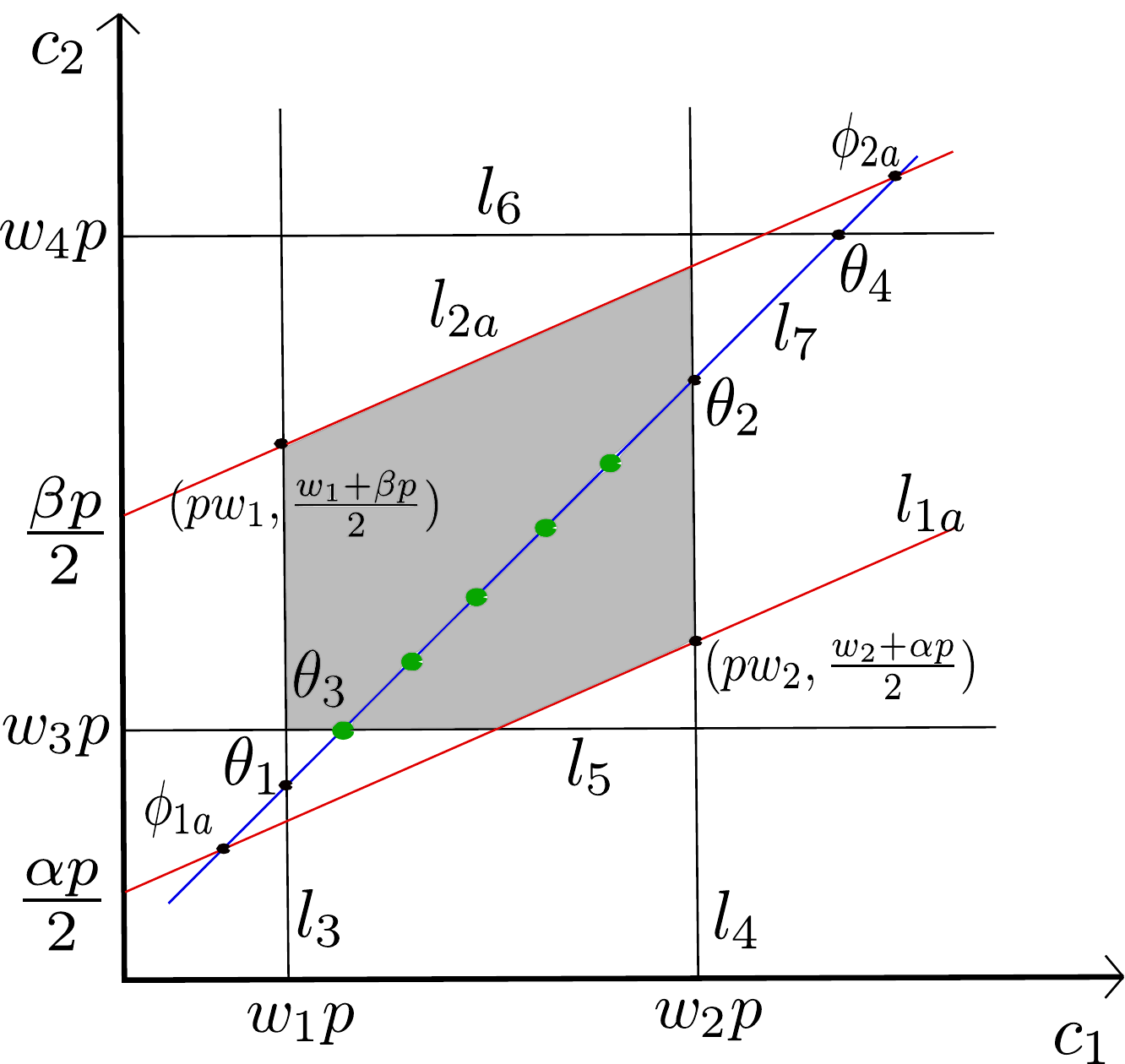}
\caption{Case $\ell=1$ in Lemma~\ref{lem:5}: the valid integer points satisfying
  (\ref{eq:2}), (\ref{eq:c1}) and (\ref{eq:c2}), are shown in green. Here,
  $\mathcal{N}_{1,\mathcal{R}}=5$; additionally, $\sigma_{1a}=\theta_3$, and
  $\sigma_{2a}=\theta_2$.}
\label{fig:N1R}
\end{figure}

\subsection{Enumeration of $(3,3)$-absorbing sets in $H(3,p,\boldsymbol{\xi},L)$ via line counting}
\label{B}

Due to the constraints given in Lemma \ref{lem:1}, there are seven possible
structures of regions in the matrix $H(3,p,\boldsymbol{\xi},L)$ in which a
6-cycle could reside \cite{AREKD16}. Let $\mathcal{R}_{1},\ldots,
\mathcal{R}_{4}$ denote the four possible regions contained within a block
column of the matrix -- that is, all variable nodes in the cycle are
contained in a single position of variable nodes, or an $H_{0}-H_{1}$ column
in the parity-check matrix. These regions may be seen in Fig.
\ref{fig:2}. 6-cycles may also be present within regions
$\mathcal{R}_5, \mathcal{R}_6, \mathcal{R}_7$, which arise when they span two block columns of the matrix.

\begin{figure}[h]
\centering
\includegraphics[scale=0.3]{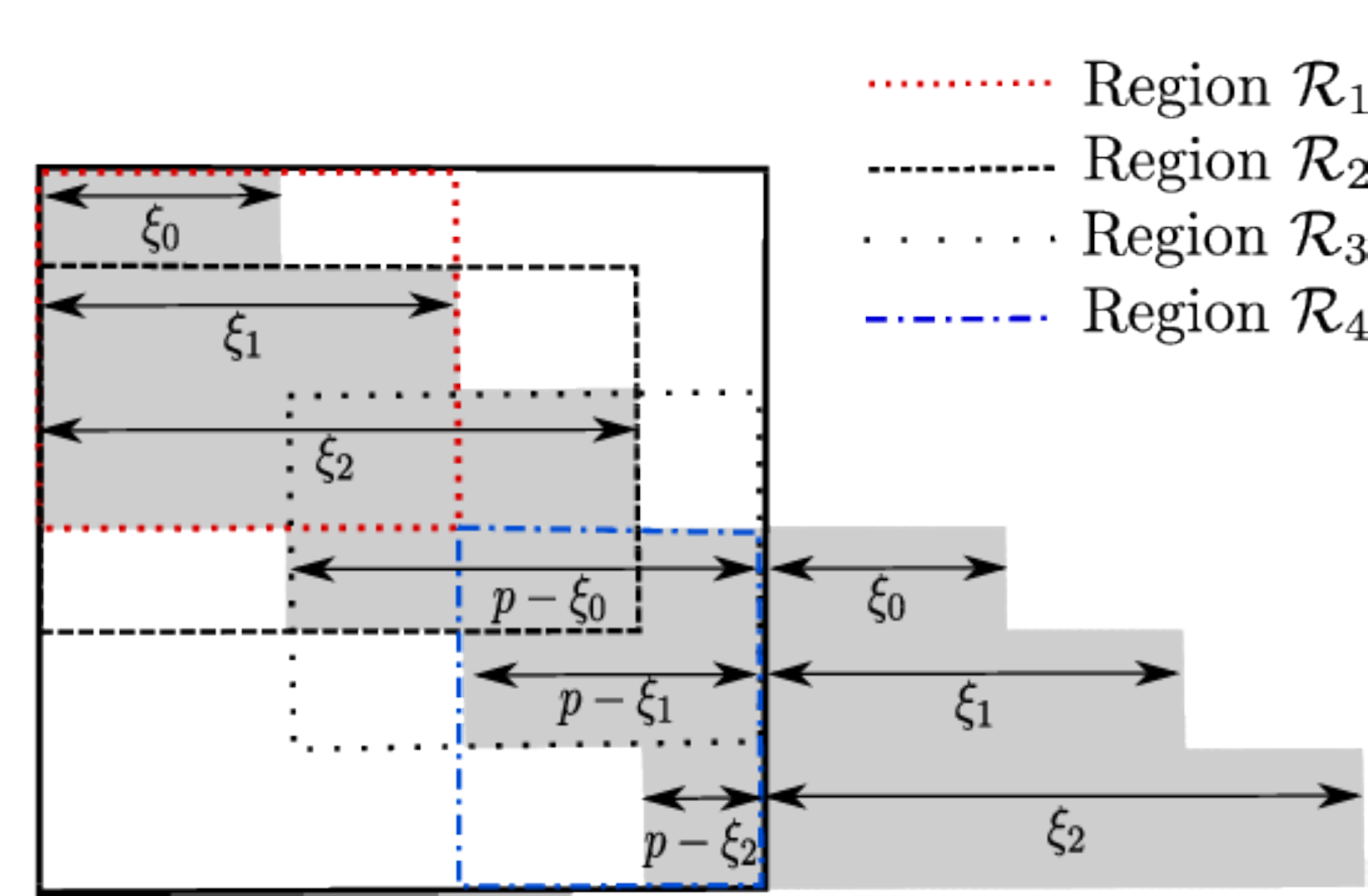}
\caption{A depiction of the structure of regions $\mathcal{R}_{1}-\mathcal{R}_{4}$ in the matrix $H(3,p,\mathbf{\mathbf{\xi}},L)$.}
\label{fig:2}
\end{figure}

\begin{lem}
\label{lem:2}
Let the number of 6-cycles in a single iteration of region types $\mathcal{R}_1-\mathcal{R}_4$ be given by $\mu_1$, and let the number of 6-cycles in a single iteration of region types $\mathcal{R}_5-\mathcal{R}_7$ be given by $\mu_2$.
Then, the total number of 6-cycles in the AB-SC-LDPC code is $L\mu_{1}+(L-1)\mu_2$.
\end{lem}
\begin{proof}
The proof follows from the structure of $H(3,p,\boldsymbol{\xi},L)$. 
\end{proof}

In particular, $\mu_{1}=\sum_{\mathcal{R}=1}^{4}\sum_{\ell=1}^{4}\mathcal{N}_{\ell,\mathcal{R}}$ and $\mu_{2}=\sum_{\ell=1}^{4}\mathcal{N}_{\ell,5}-\sum_{\ell=1}^{4}\mathcal{N}_{\ell,3}+\sum_{\ell=1}^{4}\mathcal{N}_{\ell,6}-\sum_{\ell=1}^{4}\mathcal{N}_{\ell,4}-\sum_{\ell=1}^{4}\mathcal{N}_{\ell,1}+\sum_{\ell=1}^{4}\mathcal{N}_{\ell,7}-\sum_{\ell=1}^{4}\mathcal{N}_{\ell,2}$. 

\subsection{Piecewise line counting: method for enumeration of $(3,3)$-absorbing sets in $H(3,p,L)$}
\label{C}

The line counting method discussed above can be leveraged to count $(3,3)$-ABS in $H(3,p,L)$, as well. Recall that spatially-coupled codes obtained via
graph lifting can be reordered to have the structure shown in Matrix
(\ref{tailbitingmatrix}). Due to this reordering step, the circulant
matrices belonging to $H_i$ in Matrix (\ref{tailbitingmatrix}), for
$i\in\{0,\ldots,m\}$, are no longer a contiguous share of $H(3,p)$. In
contrast to $H(3,p,\boldsymbol{\xi},L)$ there now may exist all-zero blocks
between circulants in a given block row.
As a result, the line $l_7$ that contains the valid $(c_1,c_2)$ pairs
related to the 6-cycles of Matrix (\ref{tailbitingmatrix}) ``splits," and
this splitting is contingent upon the location of the zero blocks. This
leads to disjoint regions on the $(c_1,c_2)$ plane containing the valid
integer points on $l_7$ for a given $\ell$ and $n$.
As a result, in case of $H(3,p,L)$, the length of $l_7$ in each of these regions can be found by applying a distinct set of input parameters to the line counting algorithm, leading to \textit{piecewise line counting}. The values of the elements in these sets are contingent upon the locations of the zero blocks of the region.
As a result, $\mathcal{N}_{\ell,\mathcal{R}}$ can be found in $H(3,p,L)$ for any choice of $m$ using (\ref{eq:N1R}). However, the number of regions $\mathcal{R}$ needed for enumerating $(3,3)$-ABS in $H(3,p,L)$  via piecewise line counting is greater than in $H(3,p,\boldsymbol{\xi},L)$, and this number increases significantly with $m$.

\section{Results}
\label{simulation}

In order to compare the cycle-breaking approach outlined in Section \ref{removelift} to previous 
edge-spreading methods, we consider three AB-SC-LDPC constructions. In all
three cases, the codes are obtained from the array-based $H(3,p)$
base matrix with $p=17$.  Although the lifting method of
Section \ref{removelift} extends to higher memory, the examples considered have
memory fixed at $m=1$ or $2$, as indicated. We consider:
\begin{itemize}
\item \textbf{Code 1:} This code is obtained by coupling
  $H(3,17)$ using the optimal cutting vector of \cite{AREKD16} (i.e., $m=1$). 
\item \textbf{Code 2:} This code is obtained by lifting $H(3,17)$ using the optimized $B_{m}$ matrix for the case $m=1$. 
\item \textbf{Code 3:} This code is obtained by lifting $H(3,17)$ using the optimized $B_{m}$ matrix for the case $m=2$.  
\end{itemize}

We minimize the number of $(3,3)$-ABS in each of these codes for both
windowed and non-windowed BP decoding by finding a
suitable permutation assignment matrix $B_m$. This matrix is obtained via a numerical
optimization technique: $B_m$ is optimized using a procedure combining a
limited exhaustive search with iterative backtracking. In each
step of this algorithm, the number of $(3,3)$-ABS is determined efficiently
via the line counting approach of the previous section.

Enumeration results for the non-windowed case are shown in
Table~\ref{tab:abs_count_3_codes}. This table compares the numbers
of $(3,3)$-ABS for Codes 1-3. The number of $(3,3)$-ABS in Code 3 should
compared to the numbers in the other AB-SC-LDPC codes with $m=2$: those obtained via the optimization techniques
of \cite{MR17,EHD17}.  Fig. \ref{fig:res1}
displays these results using the parameter $r_1$:
\[
r_1=\frac{\text{total \# of ABS in Code 1, 2 or 3}}{\text{total \# of ABS in the corresponding uncoupled code}}.
\]
Note that in the non-windowed case, the resulting code's block
length is equal to $Lp^{2}=289L$. The number of $(3,3)$-ABS in the uncoupled code for $p=17$ is $4624L$. Clearly, Code 3 outperforms all the other codes. 

\begin{table}[htb]
\centering
\begin{tabular}{|c|c|c|c|c|c|}
\hline
\textbf{L}&\textbf{Code 1}&\textbf{Code 2}&\textbf{Code 3}&\textbf{\cite{EHD17} for m=2}&\textbf{\cite{MR17} for m=2} \\
\hline
$10$&$19108$&$5644$&$442$&n/a&$646$\\
\hline
$20$&$39508$&$11764$&$952$&n/a&$1326$\\
\hline
$30$&$59908$&$17884$&$1462$&$4335$&$2006$\\
\hline
$40$&$80308$&$24004$&$1972$&n/a&$2686$\\
\hline
$50$&$100710$&$30124$&$2482$&n/a&$3366$\\
\hline
\end{tabular}
\caption{\label{tab:abs_count_3_codes} The number of $(3,3)$-ABS in Codes 1-3, in addition to the results
  presented in \cite{MR17,EHD17}. }
\end{table}

\begin{figure}[htb]
\centering
\includegraphics[scale=0.18]{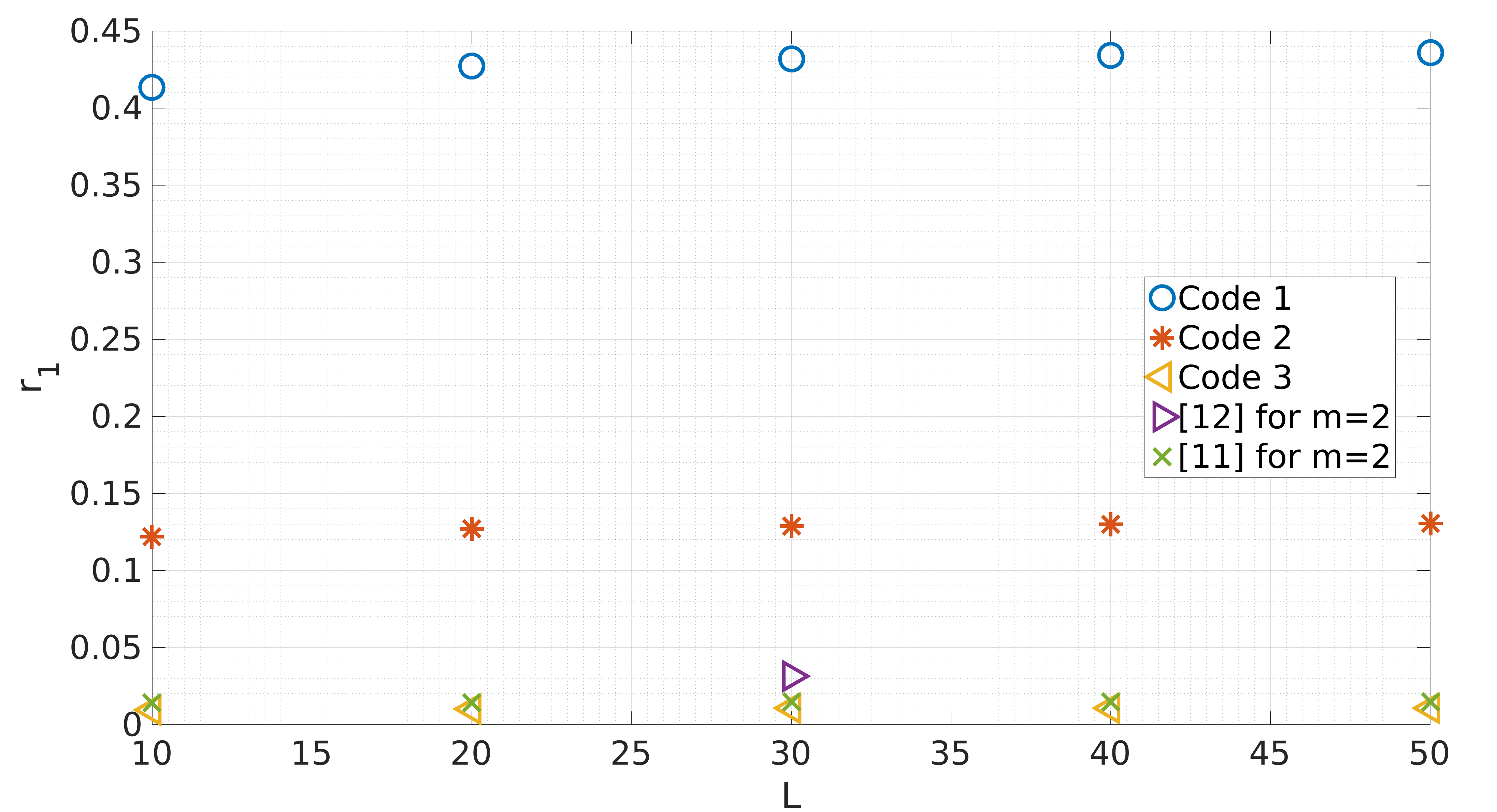}
\caption{Values of $r_1$ for various coupling lengths.
}
\label{fig:res1}
\end{figure}

Next, we present ABS counting results for a sliding windowed decoder. In
particular, we compare the number of $(3,3)$-ABS in all the sliding
positions of the windowed decoder to the number of $(3,3)$-ABS seen by the standard decoder of the
same code. The window 
in $H(\gamma,p,\boldsymbol{\xi},L)$ is positioned such that the block rows inside a window of length $p^2S_1$ variable nodes (VNs) are contained within $S_1$ contiguous $H_1-H_0$ row groups, with the total number of block rows inside the window being $\gamma (S_1-1)+1$ for $\gamma=3$ and integer $S_1\geq2$. An example of such a placement is shown in Fig.
\ref{fig: window}. The same placement technique is
applicable for the $H(3,p,L)$ code in the case that $m=1$. For $m=2$ however, the window in $H(\gamma,p,L)$ of length $p^2S_2$ VNs is contained within $S_2-2$ contiguous $H_2-H_1-H_0$ row groups such that the total number of block rows inside the window is $\gamma (S_2-3)+1$, where $S_2\geq 4$. Positioning the window in this way
ensures that all the windows sliding across the matrix are identical, and
that parity-check equations are not broken. Note that in each of these cases there are no $(3,3)$-ABS in the region where the windows overlap, since only one block row is common between two consecutive sliding positions.

\vspace{-0.5in}

\begin{figure}[tbh]
\centering{}
\includegraphics[bb=0bp 0bp 1209bp
588bp,clip,scale=0.25]{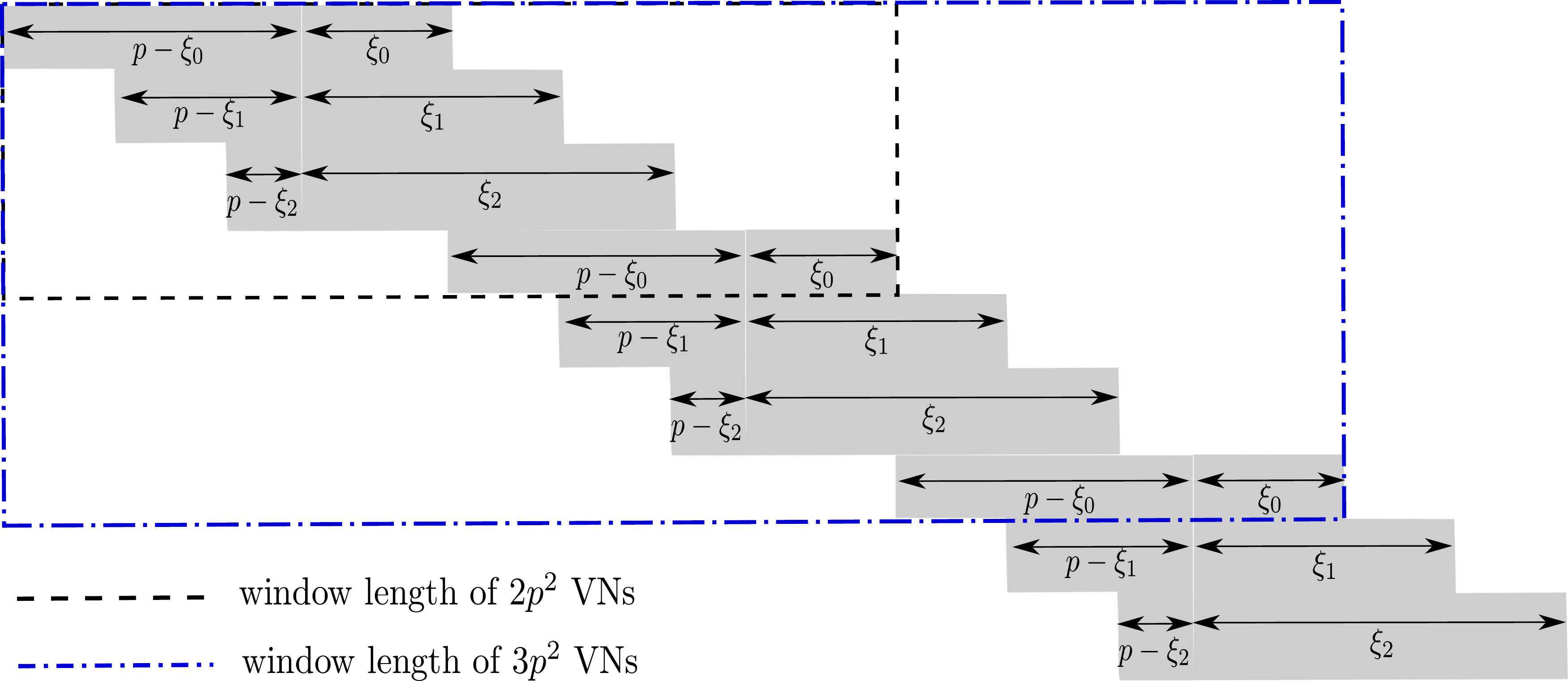}\caption{Example of window placement in $H(3,p,\boldsymbol{\xi},L)$ shown for $S_1=2,3$.} \label{fig: window}
\end{figure}

Results obtained using line counting for the BP windowed 
decoder are shown in Fig.~\ref{fig:res2} using the parameter $r_2$:
\[
r_2=\frac{\text{total \# of ABS at all sliding positions of
    window}}{\text{total \# of ABS  seen by the standard decoder}}.
\]
Table~\ref{tab:window_dec} contains the number of $(3,3)$-ABS for Codes 1, 2
and 3 with varying window sizes. It is worth noting that Code 3 for a
window length of $4p^2$ bits has no $(3,3)$-ABS at all, making it an excellent
candidate for windowed decoding.

\begin{table}[htb]
\centering
\begin{tabular}{|c|c|c|c|}
\hline
\textbf{Window Length (VNs)}&\textbf{Code 1}&\textbf{Code 2}&\textbf{Code 3} \\
\hline
$2p^2$&$1700$&$51$&n/a\\
\hline
$3p^2$&$3740$&$544$&n/a\\
\hline
$4p^2$&$5780$&$1156$&$0$\\
\hline
$5p^2$&$7820$&$1768$&$85$\\
\hline
\end{tabular}
\caption{The number of $(3,3)$-ABS in Codes 1-3 for varying window sizes.}
\label{tab:window_dec}
\end{table}
      
\begin{figure}[htb]
\centering
\includegraphics[scale=0.18]{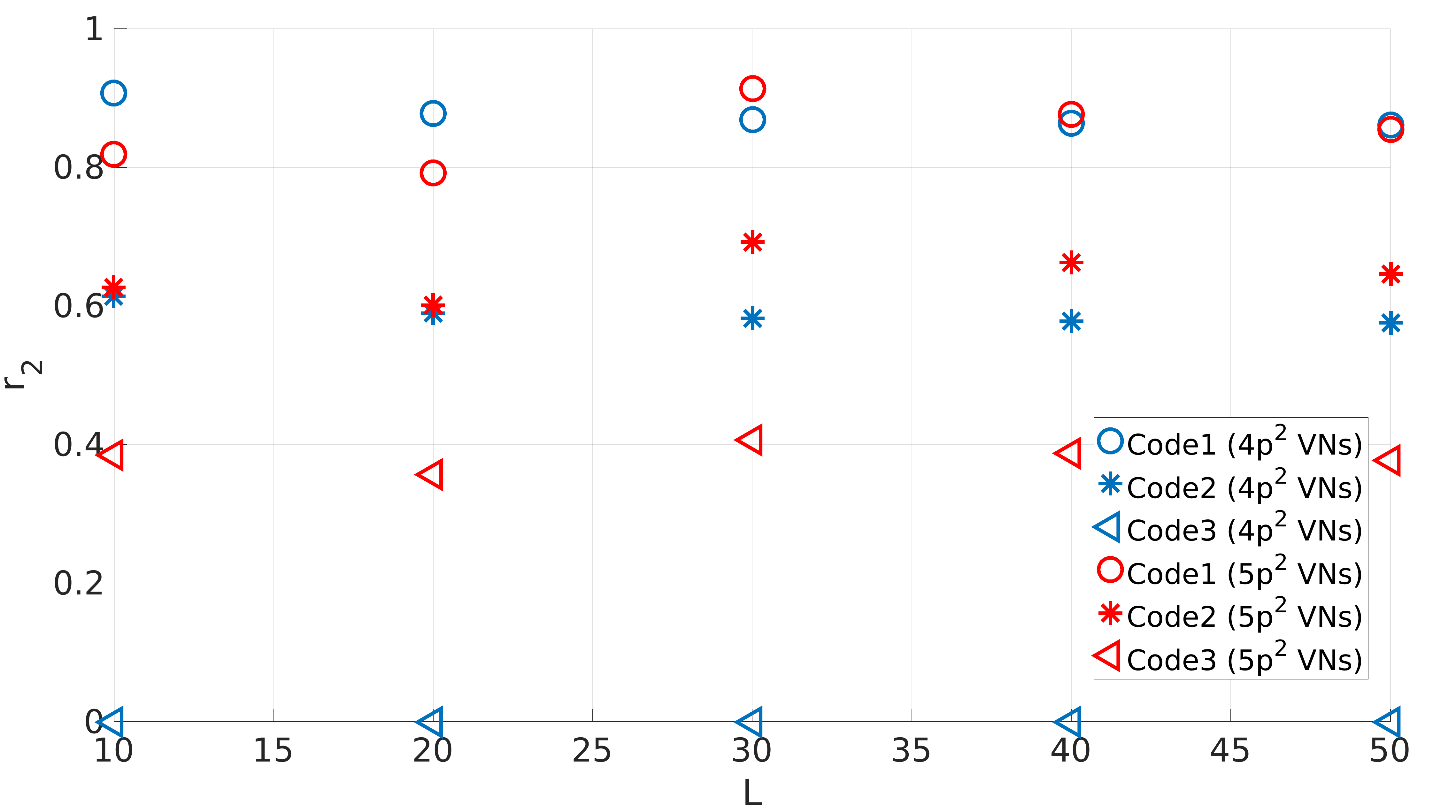}
\caption{Values of $r_2$ for various coupling lengths, with window lengths of $4p^2$ and $5p^2$.
}
\label{fig:res2}
\end{figure}

\section{Conclusion}
\label{conclusion}
We presented a generalized description of SC-LDPC codes using algebraic lifts; this framework allows for greater flexibility in code design, and for the removal of harmful absorbing sets in the design process. We introduced a novel absorbing set enumeration method and used this to demonstrate that our generalized method has the potential to outperform conventional array-based SC-LDPC construction methods such as the (generalized) cutting vector. Further optimization using multiple permutation assignments per block will be included in the full version of this paper.




%

\bibliographystyle{IEEEtran}
\bibliography{BHKK_Allerton}

\end{document}